\PassOptionsToPackage{alpine}{ifsym}
\documentclass{theoretics}

\usepackage{mdframed}
\usepackage{comment}
\usepackage{tikz, pgfplots}
\pgfplotsset{compat=1.18}
\usetikzlibrary{decorations.pathreplacing}
\usetikzlibrary{graphs, positioning}
\usepackage{tikzsymbols}
\usetikzlibrary{decorations.text}
\usepackage{animate}
\usetikzlibrary{calc}
\usetikzlibrary{arrows.meta}

\usepackage{tikzsymbols}
\usetikzlibrary{decorations.pathmorphing}
\tikzset{snake it/.style={decorate, decoration=snake}}

\definecolor{DarkRed}{rgb}{0.65,0,0}
\definecolor{Red}{rgb}{1,0,0}
\newcommand{\eps}{\varepsilon}

\newcommand{\dist}{\text{dist}}

\newcommand{\cee}{\mathcal{C}}

\newcommand{\forward}{\texttt{forward}}
\newcommand{\backward}{\texttt{backward}}

\usepackage{todonotes}

\newcommand{\Althofer}{Alth\"{o}fer}

\usepackage{booktabs}
\usepackage{longtable}
\usepackage{multirow}

\newcommand{\mst}{\texttt{MST}}

\ThCSnewtheostd{puzzle}
\ThCSnewtheostd{solution}

\title{An Alternate Proof of Near-Optimal Light Spanners}

\ThCSauthor{ Greg Bodwin}{ bodwin@umich.edu }[0000-0001-9896-8906]
\ThCSaffil{ University of Michigan, USA}
\ThCSthanks{This is an invited paper from SOSA 2024~\cite{Bodwin24}.
This work was supported by NSF:AF 2153680.}
\ThCSshortnames{G.\ Bodwin}  
\ThCSshorttitle{An Alternate Proof of Near-Optimal Light Spanners}
\ThCSyear{2025}
\ThCSarticlenum{2}
\ThCSreceived{Mar 7, 2024}
\ThCSaccepted{Nov 24, 2024}
\ThCSpublished{Jan 10, 2025}
\ThCSdoicreatedtrue
\ThCSkeywords{Extremal Graph Theory, Sparsifiers, Spanners, Lightness, Hiker Lemmas}

\begin{document}

\maketitle

\begin{abstract}
In 2016, a breakthrough result of Chechik and Wulff-Nilsen [SODA '16] established that every $n$-node graph $G$ has a $(1+\eps)(2k-1)$-spanner of lightness $O_{\eps}(n^{1/k})$, and recent followup work by Le and Solomon [STOC '23] generalized the proof strategy and improved the dependence on $\eps$.
We give a new proof of this result, with the improved $\eps$-dependence.
Our proof is a direct analysis of the often-studied greedy spanner, and can be viewed as an extension of the folklore Moore bounds used to analyze spanner sparsity.
\end{abstract}

\section{Introduction}

We study \emph{spanners}, which are a graph-theoretic primitive with applications in graph algorithms, network design, and sketching \cite{ABSHJKS20}.
\begin{definition}[Spanners \cite{PU89jacm,PU89sicomp}]
Given a graph $G$, a $t$-spanner is an edge-subgraph $H$ that satisfies
$\dist_H(u, v) \le t \cdot \dist_G(u, v)$ for all vertices $u, v$.
\end{definition}

The usual goal is to design a spanner with a favorable tradeoff between its stretch $t$ and its \emph{size}.
There are two different ways that spanner size is commonly measured.
One is by the sparsity (number of edges) of the spanner.
The stretch/sparsity tradeoff has long been understood, thanks to the following classic theorem by \Althofer{} et al.~\cite{ADDJS93}
Here and throughout the paper, all graphs are undirected and may have arbitrary positive edge weights.
\begin{theorem} [\cite{ADDJS93}] \label{thm:sotasparse}
For all positive integers $k, n$, every $n$-node graph $G$ has a $(2k-1)$-spanner $H$ on $|E(H)| \le O(n^{1+1/k})$ edges.
This tradeoff is best possible, assuming the girth conjecture \cite{girth}.
\end{theorem}

The other popular way to measure spanner size is by the \emph{total edge weight} in the spanner $w(H)$.
In general, the total edge weight required for a $t$-spanner might be unbounded, say by scaling up the edge weights of the input graph.
So in order to prove stretch/weight tradeoffs, we typically normalize the spanner weight by the weight of a minimum spanning tree (MST) of the original graph.
This size measure is called \emph{lightness}:
\begin{definition} [Spanner Lightness]
The lightness of a subgraph $H$ of a graph $G$ is the quantity
$$\ell(H \mid G) := \frac{w(H)}{w(\mst(G))}$$
where $\mst(G)$ is a minimum spanning tree of $G$.\footnote{Throughout this paper, we will assume where convenient that graphs are connected, so that $\mst(G)$ exists.  Otherwise, a minimum spanning forest may be used.}
For brevity we also write $\ell(H) := \ell(H \mid H)$.
\end{definition}
% In other words, $w(\mst(G))$ is the baseline weight that we must pay just to obtain a spanner $H$ that preserves connectivity, but with no nontrivial distance approximation.
% The lightness is the additional amount that we pay in our spanner to obtain our $\cdot t$ distance approximation.

There has been a long line of work on the stretch/lightness tradeoff; see Table \ref{tbl:priorwork} for the progression of results.
A key result in this sequence, by Chechik and Wulff-Nilsen \cite{CW18}, established the following analog of Theorem \ref{thm:sotasparse}:

\begin{theorem} [\cite{CW18}] \label{thm:sotalight}
For all $\eps>0$ and positive integers $k, n$, every $n$-node graph $G$ has a $(1+\eps)(2k-1)$-spanner $H$ of lightness
$\ell(H \mid G) = O_{\eps}\left( n^{1/k} \right).$
\end{theorem}

The stretch/lightness tradeoff in Theorem \ref{thm:sotalight} is best possible, assuming the girth conjecture \cite{girth}, and up to its dependence on $\eps$.
The theorem is proved using an ingenious framework for hierarchical graph clustering.
A followup paper by Le and Solomon \cite{LS23} refined and vastly generalized the scope of this hierarchical clustering method, improving the hidden $\eps$-dependence in the lightness bound of Theorem \ref{thm:sotalight}, and also gaining broad applications to the study of light spanners in some more specific graph classes.

\begin{table}[t]
\begin{center}
\begin{tabular}{llcc}
\textbf{Stretch} & \textbf{Lightness} & \textbf{Analyzes Greedy Spanner?} & \textbf{Citation} \\ \hline
$2k-1$ & $O(n/k)$ & \checkmark{} & \cite{ADDJS93} \\
$(1+\eps) \cdot (2k-1)$ & $O_{\eps}\left(k \cdot n^{1/k}\right)^*$ & \checkmark{} & \cite{CDNS95} \\
$(1+\eps) \cdot (2k-1)$ & $O_{\eps}\left(\frac{k}{\log k} \cdot n^{1/k} \right)^*$ & \checkmark{} & \cite{ENS15} \\
$(1+\eps) \cdot (2k-1)$ & $O\left( \eps^{-(3+2/k)} n^{1/k} \right)$ & & \cite{CW18} \\
$(1+\eps) \cdot (2k-1)$ & $O\left( \eps^{-1} n^{1/k} \right)^*$ & & \cite{LS23} \\
$(1+\eps) \cdot (2k-1)$ & $O\left( \eps^{-1} n^{1/k} \right)$ & \checkmark{} & this paper \\ \hline
$(1+\eps) \cdot (2k-1)$, $\eps = n^{-\frac{1}{2k-1}}$ only & $\Omega_k\left( \eps^{-1/k} n^{1/k} \right)$ & n/a & \cite{BF24}
\end{tabular}
\end{center}
\caption{\label{tbl:priorwork} Work on the stretch/lightness tradeoff for spanners.  Lightness bounds marked with $^*$ have an \emph{unconditionally} optimal dependence on $n$, i.e., they automatically improve if the girth conjecture fails.}
\end{table}

Meanwhile, perhaps the most popular spanner construction in the literature is the \emph{greedy algorithm} (see Algorithm \ref{alg:greedy}).
The greedy algorithm is ubiquitous because it is simple, easy to prove correct, and its stretch/sparsity and stretch/lightness tradeoffs are both known to be \emph{existentially optimal} \cite{ADDJS93, FS20}.
That is, the stretch/lightness tradeoff achieved by any algorithm -- including the hierarchical clustering method \cite{CW18, LS23} -- is automatically achieved by the greedy algorithm as well.
This has motivated interest in spanner size bounds that are proved by analyzing the greedy algorithm directly.
In the context of spanner \emph{sparsity}, there is indeed a simple proof of Theorem \ref{thm:sotasparse} that works by directly analyzing the greedy spanner.
This proof is called the Moore bounds and it is considered folklore; we recap the proof in Section \ref{sec:moore}.
In the context of spanner \emph{lightness}, there are some arguments that directly analyze the greedy spanner \cite{ADDJS93, CDNS95, ENS15}, but they all show suboptimal lightness bounds that do not quite match the one in Theorem \ref{thm:sotalight}.

The contribution of this paper is a new proof of Theorem \ref{thm:sotalight}, with the improved $\eps$-dependence from \cite{LS23}, which directly analyzes the greedy spanner (or, more accurately, which directly analyzes graphs of high weighted girth \cite{ENS15}; see Section \ref{sec:wtgirth}).
Our proof follows the template of the Moore bounds, and so it may also have an advantage in conceptual familiarity to a reader who is primarily comfortable with the literature on spanner sparsity.

\subsection{Other Related Work}

In subsequent work to the current paper, the author and Flics \cite{BF24} constructed a lower bound graph on which any $(1+\eps)(2k-1)$-spanner has lightness $\Omega_k(\eps^{-1/k} n^{1/k})$, for some particular choice of $\eps$ (as usual, this result is conditional on the girth conjecture).
This implies that the dependence of $\eps^{-1}$ in the upper bound of \cite{LS23} and the current paper cannot be \emph{completely} removed, although it remains an interesting open problem to tighten the $\eps$-dependence on either the upper or the lower bounds side.

The greedy algorithm runs in time $O(mn^{1+1/k})$.
There has been a line of work on alternate algorithms that achieve similar lightness guarantees, but which run faster.
Elkin and Solomon \cite{ES16} constructed spanners in near-linear time whose lightness is suboptimal by a factor of $k$.
Alstrup, Dahlgaard, Filtser, St\"{o}ckel, and Wulff-Nilsen \cite{ADFSW22} constructed spanners with near-optimal lightness (up to dependence on $\eps$), in time approximately $O(n^{2 + 1/k})$.
Finally, the recent work of Le and Solomon \cite{LS23} tightened the $\eps$-dependence, and achieved near-linear runtime.

Light spanners have also been studied intensively in various specific graph classes.
This most notably includes Euclidean spaces (see e.g.\ \cite{FS20, LS23, NS07, ES15} and references within), but other domains have also received attention, such as doubling metrics \cite{KLMS22, LT24, BLW19, CLNS15, CLN15, CG09, GR08}, minor-free graphs \cite{BLW17}, and polygonal domains \cite{BKKLLPT24}.

\subsection{Organization}

A lot of this paper is optional ``warmup'' content rather than the main proof, and the enterprising reader can get the full proof by reading Sections \ref{sec:priorwork} and \ref{sec:mainproof} only.
The surrounding warmup proofs, discussions, and puzzles are recommended to help build intuition.

\section{Warmup 1: The Moore Bounds for Spanner Sparsity \label{sec:moore}}

We will begin by recapping the proof of the stretch/sparsity tradeoff given in Theorem \ref{thm:sotasparse}.
This theorem analyzes the following simple \emph{greedy algorithm}:
\begin{center}
\begin{algorithm}[h]
\textbf{Input:} Graph $G = (V, E, w)$, stretch $t$\;~\\

Let $H = (V, \emptyset, w)$ be the initially-empty spanner\;
\ForEach{$(u, v) \in E$ in order of nondecreasing weight}{
    \If{$\dist_H(u, v) > t \cdot w(u, v)$}{
        add $(u, v)$ to $H$\;
    }
}
\textbf{return} $H$\;

\caption{\label{alg:greedy} The Greedy Spanner Algorithm \cite{ADDJS93}}
\end{algorithm}
\end{center}
We first observe that the output spanner $H$ of the greedy algorithm with stretch parameter $2k-1$ has girth (shortest cycle length) $>2k$ \cite{ADDJS93}.
(We omit this proof, as it is standard, but we note that it is implied by Lemma \ref{lem:wtgirth} to follow.)
Theorem \ref{thm:sotasparse} then follows from the \emph{Moore bounds}, which limit the maximum possible number of edges in a high-girth graph:

\begin{theorem} [Moore Bounds]
For any positive integers $n, k$, every $n$-node graph $H$ with girth $>2k$ has $O(n^{1+1/k})$ edges.
\end{theorem}

The proof of the Moore bounds is a counting argument over the \emph{edge-simple $k$-paths} of $H$.
Recall that an edge-simple path is one that does not repeat edges.
One part of this counting argument is the following \emph{dispersion lemma}, implying that these paths are ``dispersed'' around the graph rather than having several of them concentrated on a pair of endpoint nodes.

\begin{lemma} [Unweighted Dispersion Lemma]
$H$ may not have two distinct edge-simple $k$-path with the same endpoints $s, t$.
\end{lemma}
\begin{proof}
Suppose for contradiction that $\pi_a, \pi_b$ are two distinct $s \leadsto t$ edge-simple $k$-paths in $H$.
The subgraph $\pi_a \cup \pi_b$ is not a tree, since it contains two distinct $s \leadsto t$ paths, and so it contains a cycle $C$.
This cycle must have $|C| \le |\pi_a| + |\pi_b| = 2k$ edges, which contradicts that $H$ has girth $>2k$.
\end{proof}

The dispersion lemma implies an upper bound of $O(n^2)$ on the number of edge-simple $k$-paths in $H$.
The other part of the argument is a \emph{counting lemma}, which gives a corresponding lower bound.
Only the last ``full'' counting lemma in the following sequence is used, but the proof strategy is to bootstrap it by starting with weaker versions.

\begin{lemma} [Unweighted Weak Counting Lemma]
If $|E(H)| \ge n$, then $H$ contains an edge-simple $k$-path.
\end{lemma}
\begin{proof}
Since $|E(H)| \ge n$, $H$ contains a cycle $C$.
Since $H$ has girth $>2k$, there are $>2k$ edges in $C$.
Thus, any subpath of $C$ of length $k$ is an edge-simple $k$-path.
\end{proof}

\begin{lemma} [Unweighted Medium Counting Lemma]
$H$ contains at least $|E(H)| - n$ edge-simple $k$-paths.
\end{lemma}
\begin{proof}
Repeat the following process until no longer possible: find an edge-simple $k$-path $\pi$, record it, and then delete any \textbf{one} edge in $\pi$ from $H$ to ensure that we don't re-record $\pi$ in a future round.
By the weak counting lemma, we may repeat this process for at least $|E(H)| - n$ rounds.
\end{proof}

\begin{lemma} [Unweighted Full Counting Lemma] \label{lem:moorefc}
Let $d := |E(H)|/n$.
If $d \ge 2$, then $H$ contains $n \cdot \Omega(d)^k$ edge-simple $k$-paths.
\end{lemma}
\begin{proof}
Let $H'$ be a random edge-subgraph of $H$, obtained by keeping each edge independently with probability $2/d$.
Let $p, p'$ be the number of edge-simple $k$-paths in $H, H'$, respectively.
On one hand, for any edge-simple $k$-path $\pi$ in $H$, the probability that $\pi$ survives in $H'$ is $\Theta(d)^{-k}$, and so
$\mathbb{E}[p'] = p \cdot \Theta(1/d)^{k}$.
On the other hand, we have
\begin{align*}
\mathbb{E}[p'] &\ge \mathbb{E}[|E(H')| - n] \tag*{Medium Counting Lemma}\\
&= |E(H)| \cdot \frac{2}{d} - n\\
&= 2n - n\\
&= n.
\end{align*}
Combining these inequalities, we get $n \le p \cdot \Theta(1/d)^{k}$, and rearranging gives $p \ge n \cdot \Theta(d)^k$.
\end{proof}

We are now ready to complete the proof of the Moore bounds.
Let $d := |E(H)|/n$.
If $d < 2$ then $|E(H)| = O(n)$ and we are done.
Otherwise, if $d \ge 2$, then by the full counting lemma $H$ has $n \cdot \Omega(d)^k$ edge-simple $k$-paths.
Meanwhile, the dispersion lemma implies that $H$ has $O(n^2)$ edge-simple $k$-paths.
Comparing these estimates, we get
$$n \cdot \Omega(d)^k \le O(n^2).$$
Rearranging terms in this inequality, we get
$d \le O(n^{1/k})$,
and so $|E(H)| = O(n^{1+1/k})$.

\section{Lightness Reductions from Prior Work \label{sec:priorwork}}

\subsection{The Weighted Girth Framework \label{sec:wtgirth}}

In our previous proof of the Moore bounds, the first step is to observe that the output spanner of the greedy algorithm has high girth.
For lightness, Elkin, Neiman, and Solomon \cite{ENS15} formalized the analogous method, which we will use in this paper.

\begin{definition} [Normalized Weight and Weighted Girth \cite{ENS15}]
For a cycle $C$ in $G$, we define its normalized weight to be
$$w^*(C) := \frac{w(C)}{\max_{e \in C} w(e)}.$$
The weighted girth of $G$ is the minimum value of $w^*(C)$ over all cycles $C$ in $G$.
\end{definition}

\begin{lemma} [\cite{ENS15}] \label{lem:wtgirth}
The greedy algorithm with parameter $t$ returns a graph $H$ with weighted girth $>t+1$.
\end{lemma}
\begin{proof}
Let $C$ be a cycle in $G$ of normalized weight $w^*(C) \le t+1$.
It suffices to argue that not all edges of $C$ will be added to $H$.
Let $(u, v)$ be the last edge in $C$ considered by the greedy algorithm, and suppose that all previous edges in $C$ were added to $H$.
Then there is a $u \leadsto v$ path in $H$, through the other edges in $C$, of total weight
\begin{align*}
w(C) - w(u, v) \le (t+1) \cdot w(u, v) - w(u, v) = t \cdot w(u, v).
\end{align*}
This inequality means that the greedy algorithm will reject the edge $(u, v)$, rather than adding it to $H$, and so the cycle $C$ does not survive in $H$.
\end{proof}

Also note that the greedy algorithm essentially contains a run of Kruskal's algorithm within it \cite{CLRS22}, and so the output spanner $H$ contains an $\mst$ of $G$.
This means that instead of bounding $\ell(H \mid G)$, we may equivalently bound $\ell(H)$ over all graphs $H$ of weighted girth $>(1+\eps)2k$.

\subsection{Reduction to Unit-Weight Spanning Cycles \label{sec:utmst}}

It will be convenient in the main proof to reduce to the setting where the spanner $H$ has a very particular structure for its $\mst$:

\begin{definition} [Unit-Weight Spanning Cycles]
We say that a cycle $\cee$ in $H$ is a unit-weight spanning cycle if $\cee$ is Hamiltonian (it contains each node exactly once), all edges in $\cee$ have weight $1$, and all edges in $E(H)$ have weight $\ge 1$.
\end{definition}
So, for example, any tree created by deleting any one edge from a unit-weight spanning cycle $\cee$ is an $\mst$.
We will reduce to the case where $H$ has a unit-weight spanning cycle; the parts of this reduction all appear implicitly or explicitly in prior work \cite{CW18, LS23, ENS15}.

\begin{lemma}
Let $H$ be a non-forest $n$-node graph with weighted girth $>t$ and lightness $\ell$.
Then there exists a graph $H'$ with $O(n)$ nodes, weighted girth $>t$, lightness $\Omega(\ell)$, and a unit weight spanning cycle.
\end{lemma}
\begin{proof}
We split the reduction into two steps: first we modify to a unit-weight $\mst$, and then we change the $\mst$ into a spanning cycle.

\paragraph{Reduction to Unit-Weight $\mst$.}

\begin{itemize}
\item Rescale the edge weights of $H$ so that the average edge weight in $\mst(H)$ is $1$.
Rescaling affects neither the weighted girth nor the lightness of $H$.

\item For all edges $e \in \mst(H)$ of weight $w(e) > 1$, add new nodes to subdivide $e$ into a path of $\lceil w(e) \rceil$ edges.
Each new edge is assigned weight $w(e) / \lceil w(e) \rceil$.
Notice that:
\begin{itemize}
\item The total weight in $\mst(H)$ is $n-1$ after the previous rescaling step, and therefore we add at most $n-1$ new nodes to $H$ in this step.

\item Additionally, the weighted girth of $H$ does not change as we subdivide edges.
This holds since for each cycle $C$ the total weight of $C$ is unchanged, and also $C$ must have a heaviest edge that is not in $\mst(H)$; since this heaviest edge is not subdivided, the normalized weight $w^*(C)$ remains the same.
\end{itemize}

\item Finally, for all edges $e \in E(H)$ of weight $< 1$, increase $w(e)$ to $1$ (this step is applied to both $\mst$ and non-$\mst$ edges).
Notice that:
\begin{itemize}
\item By the previous step every edge in $\mst(H)$ has weight in the range $[1/2, 1]$, and so this step increases $w(\mst(H))$ by at most a factor of $2$, and

\item The weighted girth of $H$ is nondecreasing in this step.
This is because for each cycle $C$, $w(C)$ is increasing, and moreover: (1) if $C$ has an edge of weight $\ge 1$ then its heaviest edge weight does not change, and (2) if all edges in $C$ have weight $\le 1$, then after this step all edges have weight $1$, so its normalized weight is $w^*(C) = |C| > t$. 
\end{itemize}
\end{itemize}

\paragraph{Reduction to Spanning Cycle.}

Next, we reduce to the setting where $H$ has a spanning cycle~$C$.

\begin{figure}[h]
\begin{center}
	\includegraphics[width=\textwidth]{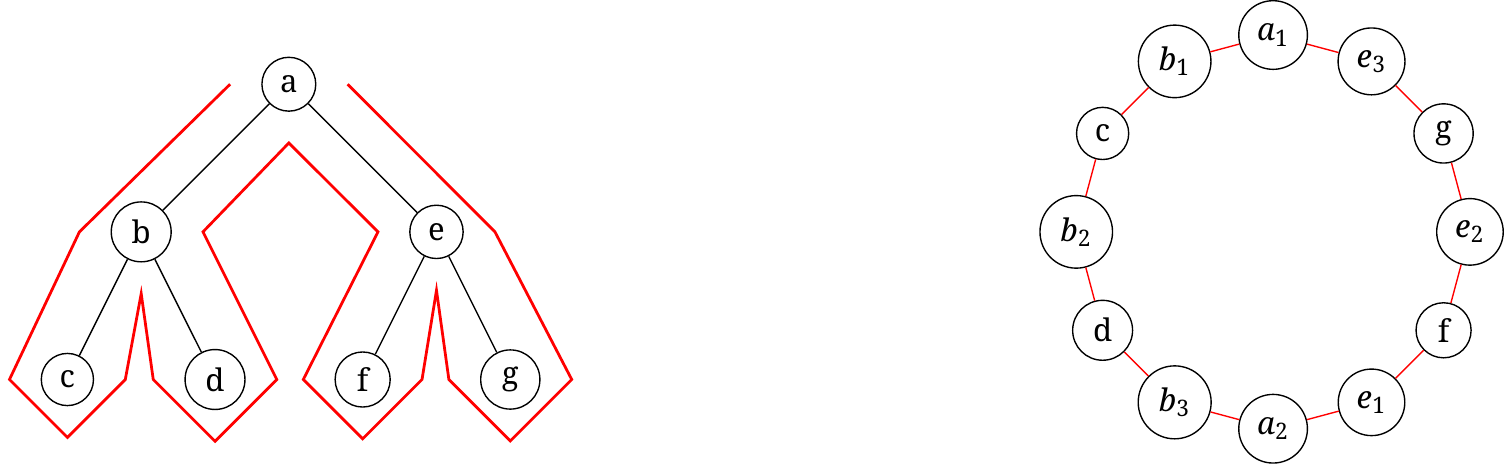}
\end{center}
\caption{(Left) For the pictured $\mst(H)$, the red line around the outside of the tree traces the tour $T = (a, b, c, b, d, b, a, e, f, e, g, e, a)$.  (Right) We construct $H'$ by mapping the tour $T$ to a spanning cycle on $2n-2$ nodes, making copies of nodes from $H$ as needed.}
\end{figure}

A \emph{tour} $T$ of $\mst(H)$ is a circularly-ordered sequence of nodes, with repeats, of the form $T = (v_0, v_1, \dots, v_{2n-2}=v_0)$, which is the node sequence of a closed walk on $\mst(H)$ that uses every edge exactly twice with opposite orientations.
Fix a tour $T$ of $\mst(H)$, and then construct $H'$ as follows:
\begin{itemize}
\item A tour of an $n$-node tree always contains exactly $2n-2$ nodes.
We will take these $2n-2$ nodes as the vertex set of $H'$; that is, some nodes in $H$ have several copies in $H'$.

\item The tour $T$ will be the spanning cycle of $H'$, meaning that for each pair of adjacent nodes along $T$, we include the corresponding edge in $E(H')$ with weight $1$.

\item For each non-spanning-cycle edge $(u, v) \in E(H \setminus T)$, we choose an arbitrary copy $u_i, v_j \in V(H')$ of $u, v$ respectively.
Then we include $(u_i, v_j) \in E(H')$ with the same weight as $(u, v)$.
\end{itemize}

It is immediate from the construction that $H'$ has $O(n)$ nodes and a unit-weight spanning cycle, and that $w(H') \ge w(H)$, and so lightness only changes by a constant factor.
Moreover, we only create one new cycle when we move from $H$ to $H'$, which is the spanning cycle $T$ itself.
We have $w^*(T) = 2n-2$, and since we have assumed that $H$ is not a forest, it has a cycle and therefore its weighted girth is $\le n$.
It follows that the weighted girth of $H'$ is at least as large as the weighted girth of $H$.
\end{proof}

One of the advantages of reducing to the case where $H$ has a unit-weight spanning cycle is that we can limit its maximum edge weight:
\begin{lemma} \label{lem:maxwt}
Let $H$ be an $n$-node graph with weighted girth $>t$ and a unit weight spanning cycle $\cee$.
Then all edges in $H$ have weight $< \frac{n}{2(t-1)}$.
\end{lemma}
\begin{proof}
Consider an edge $(u, v)$, and consider the cycle formed by $(u, v)$ and the shorter $u \leadsto v$ path through the spanning cycle, which uses at most $n/2$ edges.
The normalized weight of this cycle is at most
$$\frac{w(u, v) + n/2}{w(u, v)} = 1 + \frac{n}{2w(u, v)}.$$
By the weighted girth of $H$, this quantity must be $>t$.
Rearranging, we get
%\begin{align*}
$w(u, v) < \frac{n}{2(t-1)}$.% \tag*{\qedhere}
%\end{align*}
\end{proof}

\section{Warmup 2: Lightness Bounds via Monotone Paths \label{sec:warmup}}

We will next prove a weaker version of our main result, with an additional suboptimal $k$ factor, in order to introduce some of our new proof ideas.
\begin{theorem} [Warmup] \label{thm:warmupmain}
Let $\eps>0$, let $k, n$ be positive integers, and let $H$ be an $n$-node graph with a unit-weight spanning cycle $\cee$ and weighted girth $>(1+2\eps) \cdot 2k$.
Then
$$w(H) = O\left( \eps^{-1} k n^{1+1/k} \right).$$
\end{theorem}

The $2\eps$ term in the weighted girth, rather than $\eps$, is purely for convenience in the analysis to follow; by reparametrizing $\eps \gets \eps/2$ it does not affect the theorem statement.
We will also assume for convenience that all non-spanning-cycle edges in $H$ have distinct weights; if not, any tiebreaking method will work, e.g.\ the lexicographically smaller edge is considered lighter.
Finally, we arbitrarily choose one direction around the spanning cycle $\cee$ to be the \forward{} direction, and the reverse to be \backward{}.

\subsection{Edge-Safe Paths and Warmup Proof Intuition}

Our first step is to define the kind of paths that will be the focus of our Moore-bound-like counting argument.
The following kind of path is analogous to a \emph{single edge} in the Moore bounds; we will ultimately focus on paths made of $k$ concatenated copies of this basic building block.

\begin{definition} [Edge-Safe Paths]
A path $\pi$ is \emph{safe} for an edge $(u, v)$ if, for some integer $0 \le s \le \eps w(u, v)$, it has the following structure: it starts with a prefix of exactly $s$ \forward{} spanning cycle edges, then it uses the edge $(u, v)$, and then it ends with a suffix of exactly $s$ \backward{} spanning cycle edges.
We will say that $\pi$ is \emph{extra-safe} for $(u, v)$ if $s \le \eps w(u, v)/2$.
\end{definition}

\begin{figure}[h]
\begin{center}
\begin{tikzpicture}

   \def\radius{1.5cm} % Define the radius of the circle

    % Draw the gray circle
    \draw[gray] (0,0) circle (\radius);

    % Draw the red arc
    \fill[red] (60:\radius) circle (2mm) node [above=0.2] (u) {$u$};
    \draw[red, ultra thick, ->] (0:\radius) arc (0:60:\radius) node [midway, right] {\small $s \le \eps w(u, v)$};

    % Draw the blue arc
    \fill[red] (270:\radius) circle (2mm) node [below=0.2] (v) {$v$};;
    \draw[red, ultra thick, ->] (270:\radius) arc (270:210:\radius) node [midway, below left] {\small $s \le \eps w(u, v)$};
    % Draw the black point at the middle of the blue arc
    
    \draw [ultra thick, red] (u) to (v);

\end{tikzpicture}
\end{center}
\caption{A path that is safe for the edge $(u, v)$ (considering the counterclockwise direction around the spanning cycle to \forward{}, and clockwise to be \backward{})}
\end{figure}
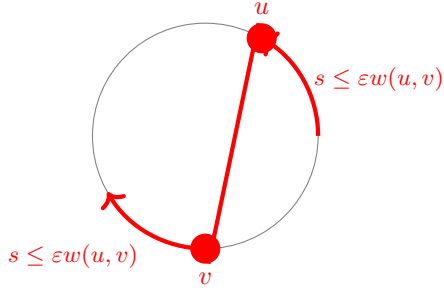

The detail of extra-safety in this definition can be ignored for now; it will be useful later for a minor technical reason in the counting lemma.
The requirement that an edge-safe path uses the same number of \forward{} and \backward{} spanning cycle edges enables the following simple yet important technical claim:

\begin{claim} \label{clm:esmatching}
Let $q, q'$ be paths in $H$ that are safe for edges $e, e'$ respectively, and which share an endpoint node $y$.
If $q \ne q'$, then $e \ne e'$.
\end{claim}
\begin{proof}
We prove the contrapositive.
Suppose $e = e' =: (u, v)$.
By Lemma \ref{lem:maxwt} we have $w(u, v) < n/2$, and so no path safe for $(u, v)$ can use more than $n$ spanning cycle edges.
Let $0 \le s < n$ be the number of \backward{} steps along the spanning cycle from $v$ to $y$.
Then both $q, q'$ must end with a suffix of exactly $s$ \backward{} steps along the spanning cycle.
So they also begin with a prefix of exactly $s$ \forward{} steps along the spanning cycle, ending at $u$, implying equality.
\end{proof}

In the same way that the Moore bounds focus on paths made up of $k$ edges, a natural proof attempt to bound lightness would be to focus on paths made up of $k$ edge-safe subpaths:

\begin{definition} [Safe $k$-Paths]
A path $\pi$ in $H$ is a \emph{safe $k$-path} if it can be partitioned into $k$ subpaths $\pi = q_1 \circ \dots \circ q_k$, where each path $q_i$ is safe for an edge $e_i$.
We say that $\pi$ is an \emph{extra-safe $k$-path} if each path $q_i$ is extra-safe for $e_i$.
\end{definition}

Unfortunately, this natural attempt breaks.
Specifically, the dispersion lemma fails: it is possible to have two edge-simple safe $k$-paths that share endpoints, without implying that $H$ has a cycle of small normalized weight.
This specifically occurs in the following figure:

\begin{figure}[ht]
\begin{center}

	\includegraphics[width=0.4\textwidth]{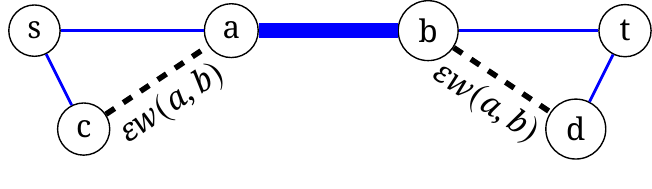}
\end{center}
\caption{\label{fig:girthfail} A counterexample to the dispersion lemma for safe $k$-paths}
\end{figure}

In Figure \ref{fig:girthfail}, the solid blue lines represent non-spanning-cycle edges, and we imagine that $(a, b)$ is much heavier than the other four such edges.
The dashed lines represent spanning cycle subpaths of length $\eps w(a, b)$ each.  This graph contains two different safe $3$-paths with endpoints $s, t$, which are
$$\pi_1 = (s, a, b, t), \pi_2 = (s, c, \dots, a, b, \dots, d, t).$$
However, neither of the two cycles
$$C_1 = (s, a, c, \dots, s), C_2 = (t, b, \dots, d, t)$$
necessarily have small normalized weight.
For example, the normalized weight of $C_1$ is
$$w^*(C_1) = \frac{w(s, c) + w(s, a) + \eps w(a, b)}{\max\{w(s, c), w(s, a)\}}.$$
This quantity can be unbounded, specifically if $w(a, b) \gg \max\{w(s, c), w(s, a)\}$.

The conceptual point made by \cite{CDNS95, ENS15} is essentially that the counterexample in Figure \ref{fig:girthfail} can be avoided if we narrow our focus to paths that do not have one edge $(a, b)$ that is much heavier than the others.
Indeed, the dispersion lemma holds for safe $k$-paths whose non-spanning-cycle edge weights differ by at most a factor of $2$ \cite{CDNS95}, or even by a factor of $k$ \cite{ENS15}.
Our following warmup proof makes the conceptual point that this counterexample can also be avoided if we narrow our focus to ``monotone'' paths, in which the non-MST edge weights may differ by a lot, but they must be increasing along the path.

\subsection{Monotone Safe Paths and the Dispersion Lemma}

Recall in the following definition that we have assumed for convenience that the non-spanning-cycle edges of $H$ have distinct weights.

\begin{definition} [Monotone Safe $k$-Paths]
Let $\pi$ be a safe $k$-path, which can be partitioned into subpaths
$\pi = q_1 \circ \dots \circ q_k$
where each $q_i$ is safe for an edge $e_i$.
We say that $\pi$ is \emph{monotone} if these edges are increasing in weight, that is, $w(e_1) < \dots < w(e_k)$.
\end{definition}

%We will now prove our dispersion lemma, showing that this restriction works as intended:

\begin{lemma} [Monotone Dispersion Lemma] \label{lem:warmupdisp}
$H$ may not have two distinct monotone safe $k$-paths with the same endpoints $s, t$.
\end{lemma}
\begin{proof}
Seeking contradiction, let $\pi^a, \pi^b$ be two distinct monotone safe $k$-paths with the same endpoints $s, t$.
Let the decomposition of $\pi^a$ into safe paths be $\pi^a = q^a_1 \circ \dots \circ q^a_k$, where each subpath $q^a_i$ is safe for the edge $e^a_i$.
We use similar notation for $\pi^b$.

Let $j$ be the largest index for which $q^a_j \ne q^b_j$.
Since the subpaths following $q^a_j, q^b_j$ are equal, $q^a_j, q^b_j$ must share an endpoint node, which we will call $y$.
By Claim \ref{clm:esmatching}, $q^a_j, q^b_j$ are safe for distinct edges $e^a_j \ne e^b_j$.
Assume without loss of generality that $w(e^a_j) > w(e^b_j)$.
By monotonicity, it follows that $e^a_j \notin \pi_b[s \leadsto y]$.
We can therefore find a cycle
$$C \subseteq \pi^a [s \leadsto y] \cup \pi^b [s \leadsto y]$$
in which $e^a_j$ is the heaviest edge.
We can bound the normalized weight of $C$ as:

\begin{figure}[h]
  \begin{center}
    	\includegraphics[width=0.4\textwidth]{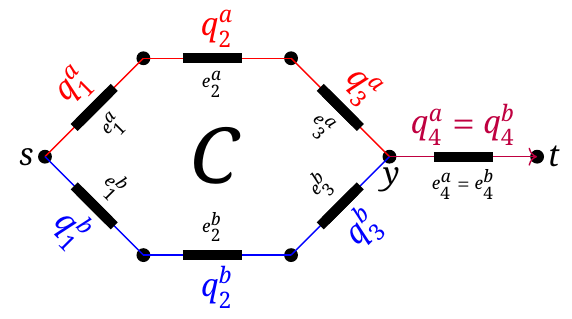}

\end{center}
\caption{The setup in our proof of the monotone dispersion lemma; each colorful segment is an edge-safe path.  The proof in this case takes $j=3$ and argues that $w^*(C) \le (1+2\eps) \cdot 6$.}
\end{figure}

\begin{align*}
w^*(C) &\le \frac{\sum \limits_{i=1}^j w(q^a_i) + w(q^b_i)}{w(e^a_j)}\\
&\le \frac{\sum \limits_{i=1}^j (1+2\eps) \left(w(e^a_i) + w(e^b_i)\right)}{w(e^a_j)}\\
&\le \frac{(j)(1+2\eps)\left( 2w(e^a_j)\right)}{w(e^a_j)}\\
&\le (1+2\eps) \cdot 2j.
\end{align*}
Since $j\le k$, this contradicts that $H$ has weighted girth $>(1+2\eps) \cdot 2k$, completing the proof.
\end{proof}

\subsection{Hiker Paths and the Counting Lemma}

The following is a famous puzzle in graph theory:

\begin{puzzle}[The Hiker Puzzle]\label{puzzle:hiker}
  We are vacationing in Graph National Park, which has $n$ landmarks (nodes) and $m$ trails (undirected edges) connecting pairs of landmarks.
Each trail has a difficulty rating (weight).
We would like to hike as many trails as possible, without repeating any trails.
However, we will get increasingly tired as we hike, and so we are only willing to hike trails in nonincreasing order of difficulty.
That is, after we hike a trail $t_i$, our next trail $t_{i+1}$ must depart from the endpoint of $t_i$ and its difficulty rating may not be higher than $t_i$'s difficulty rating.
\textbf{Prove:} If we start at the right landmark, then we can hike at least $2m/n$ trails.

\begin{center}
\begin{tikzpicture}

% Define the points for the path
\coordinate (A) at (0,0);
\coordinate (B) at (3,1);
\coordinate (C) at (5,0);
\coordinate (D) at (7,2);
\coordinate (E) at (9,1);

\node [left] at (A) {\Strichmaxerl[3]};
\node [above] at (B) {\Summertree[4]};
\node [below=0.6] at (C) {\Huge \color{brown} \VarMountain};
\node [above] at (D) {\Huge \color{red} \VarFlag};
\node [right] at (E) {\Huge \color{green!50!black} \Tent};

% Draw the lines
\draw [thick] (A) -- node[above, sloped] {12} (B);
\draw [thick] (B) -- node[above, sloped] {9} (C);
\draw [thick] (C) -- node[above, sloped] {8} (D);
\draw [thick, ->] (D) -- node[above, sloped] {4} (E);

% Optionally, draw the points
\fill (A) circle[radius=2pt];
\fill (B) circle[radius=2pt];
\fill (C) circle[radius=2pt];
\fill (D) circle[radius=2pt];
%\fill (E) circle[radius=2pt];

\end{tikzpicture}

~

A valid hiker path, $4$ trails long.
\end{center}
\end{puzzle}

\vfill
\begin{center}
  (This space is intentionally left free to not spoil the solution to Puzzle~\ref{puzzle:hiker}.)
\end{center}
\vfill
\clearpage

\begin{solution}
This problem, like so many others in life, is solved by asking our friends for help.
We invite friends to join our vacation until we have gathered a group of $n$ total hikers, and then we start with one hiker standing at each landmark.
Consider the trails $t = (u, v)$ in descending order of difficulty, and ask the hiker currently standing at $u$ and the hiker currently standing at $v$ to hike the trail, switching places with each other.

\begin{center}
\begin{minipage}{.3\textwidth}
\begin{tikzpicture} [scale=0.8]
\foreach \place/\name/\col in {{(0,0)/A/red}, {(1,2)/B/blue}, {(3,3)/D/yellow!80!black}, {(3, 0)/E/orange}} {
    \draw[fill=black] \place circle (0.2) node [above=0.2, color=\col] {\Strichmaxerl[3]};
}

\draw (5.3,-1) -- (5.3,5);
\end{tikzpicture}
\end{minipage}
\hfill
\begin{minipage}{.3\textwidth}
\begin{tikzpicture} [scale=0.8]

\foreach \place/\name/\col in {{(0,0)/A/blue}, {(1,2)/B/red}, {(3,3)/D/yellow!80!black}, {(3, 0)/E/orange}} {
    \draw[fill=black] \place circle (0.2) node [above=0.2, color=\col] {\Strichmaxerl[3]};
}

\draw [thick] (0, 0) -- (1, 2) node [midway, below=0.1, sloped] {hardest};
\draw [ultra thick, red, ->] (0, 0) to[bend left=20] (0.85, 1.9);
\draw [ultra thick, blue, ->] (1, 2) to[bend left=20] (0.15, 0.1);

\draw (5.3,-1) -- (5.3,5);
\end{tikzpicture}
\end{minipage}
\hfill
\begin{minipage}{.3\textwidth}
\begin{tikzpicture} [scale=0.8]
\foreach \place/\name/\col in {{(0,0)/A/orange}, {(1,2)/B/red}, {(3,3)/D/yellow!80!black}, {(3, 0)/E/blue}} {
    \draw[fill=black] \place circle (0.2) node [above=0.2, color=\col] {\Strichmaxerl[3]};
}

\draw [thick] (0, 0) -- (1, 2);
\draw [thick] (0, 0) -- (3, 0) node [below=0.15, midway, align=center] {second-hardest};

% \draw [ultra thick, red, ->] (0, 0) to[bend left=10] (0.85, 1.9);
% \draw [ultra thick, blue] (1, 2) to[bend left=10] (0.15, 0.1);

\draw [ultra thick, blue, ->] (0, 0) to[bend left=20] (2.8, 0);
\draw [ultra thick, orange, ->] (3, 0) to [bend left=20] (0.2, 0);

\end{tikzpicture}
\end{minipage}
\end{center}

In total, our $n$ hikers will hike $2m$ trails, and so there must exist a hiker who hiked a path of length at least $2m/n$.
This hiker hiked their trails in descending order of difficulty, and so their path satisfies the puzzle.
\end{solution}

We will now prove counting lemmas for monotone safe $k$-paths.
The medium and full counting lemmas from the Moore bounds generalize easily from Section \ref{sec:moore}, but the weak counting lemma needs a new idea, which is inspired by the hiker puzzle:

\begin{lemma} [Warmup Weak Counting Lemma] \label{lem:warmupwc}
If $w(H \setminus \cee) \ge \eps^{-1} kn$, then $H$ contains a monotone extra-safe $k$-path.
\end{lemma}
\begin{proof}
Start by placing a hiker at each node of $H$.
Then, consider the non-spanning-cycle edges of $H$ in order of increasing weight.
When an edge $(u, v)$ is considered, for each path $\pi$ that is extra-safe for $(u, v)$, we ask the two hikers at either endpoint of $\pi$ to hike $\pi$, thus switching places with each other.\footnote{A technical detail here is that each hiker walks at most one path $\pi$ per edge $(u, v)$.  This follows easily from Lemma \ref{lem:maxwt}, establishing that $w(u, v) < n/2$ and so any two extra-safe paths for $(u, v)$ have disjoint endpoints.}

% We note that by Lemma \ref{lem:maxwt}, we have $w(u, v) < n/2$, which implies that all paths that are extra-safe for $(u, v)$ have distinct endpoints; that is, no hiker will hike $(u, v)$ more than once.
% We note that the reverse of an extra-safe path is itself an extra-safe path, and so all hikers walk an extra-safe path in each step.
% We also note that no hiker hikes $(u, v)$ twice; that is, 

There are at least $\eps w(u, v)/2$ paths that are extra-safe for each edge $(u, v)$, and two hikers hike each such path (one in each direction).
Thus, after all non-spanning-cycle edges of $H$ are considered, in total our $n$ hikers have hiked at least $\eps w(H \setminus \cee) \ge kn$ extra-safe paths, and the path walked by each hiker is a monotone extra-safe path.
Thus, there exists a hiker who hiked a monotone extra-safe path of length at least $k$.
\end{proof}

The bootstrapping process from the weak to the medium and full counting lemmas essentially works exactly as in the Moore bounds, with minor adjustments.

\begin{lemma} [Warmup Medium Counting Lemma] \label{lem:warmupmed}
$H$ contains at least $\Theta(\eps) \cdot \left( w(H \setminus \cee) - \eps^{-1} kn\right)$ monotone safe $k$-paths.
\end{lemma}
\begin{proof}
Repeat the following process until no longer possible.
Find a monotone extra-safe $k$-path $\pi$, with decomposition
$$\pi = q_1 \circ \dots \circ q_k$$
where each $q_i$ is extra-safe for some edge $e_i$.
Notice that, for any integer $0 \le s \le \eps w(e_1)/2$, we can record a (not-necessarily-extra-)safe $k$-path by modifying $\pi$ by adding $s$ additional \forward{} spanning cycle edges to the start of every path $q_i$, and also adding $s$ additional \backward{} spanning cycle edges to the end of every path $q_i$.
We record $\Theta(\eps w(e_1))$ monotone safe $k$-paths in this way.

We then delete the edge $e_1$, to ensure that we do not re-record any of these paths in a future round.
By the weak counting lemma, we may repeat this process at least until $w(H \setminus \cee) < \eps^{-1} kn$.
It follows that we will record at least $\Theta(\eps) \cdot \left(w(H \setminus \cee) - \eps^{-1} kn\right)$ monotone safe $k$-paths before halting.
\end{proof}

\begin{lemma} [Warmup Full Counting Lemma] \label{lem:warmupfc}
Let $d := w(H \setminus \cee) / n$.
If $d \ge 2k \eps^{-1}$, then $H$ contains at least $kn \cdot \Omega(\eps d/k)^k$ monotone safe $k$-paths.
\end{lemma}
\begin{proof}
Let $H'$ be a random edge-subgraph of $H$ obtained by keeping the spanning cycle $\cee$ deterministically, and keeping each non-spanning-cycle edge independently with probability $2k\eps^{-1} / d$.
Let $p, p'$ be the number of monotone safe $k$-paths in $H, H'$, respectively.
On one hand, monotonicity implies that any monotone safe $k$-path $\pi$ uses $k$ \emph{distinct} non-spanning-cycle edges.
Thus, the probability that $\pi$ survives in $H'$ is $\Theta(k\eps^{-1} / d)^{k}$, and so
$$\mathbb{E}[p'] = p \cdot \Theta\left(\frac{k\eps^{-1}}{d}\right)^{k}.$$
On the other hand, we have
\begin{align*}
\mathbb{E}[p'] &\ge \mathbb{E}\left[\Theta(\eps) \cdot \left( w(H' \setminus \cee) - \eps^{-1} kn \right)\right] \tag*{Medium Counting Lemma}\\
&= \Theta(\eps) \cdot \left(\mathbb{E}\left[ w(H' \setminus \cee)\right] - \eps^{-1} kn\right)\\
&= \Theta(\eps) \cdot \left( w(H \setminus \cee) \cdot \frac{2k \eps^{-1}}{d} - \eps^{-1} kn\right)\\
&= \Theta(\eps) \cdot \left( 2kn \eps^{-1} - \eps^{-1} kn\right)\\
&= \Theta(\eps) \cdot \left( \eps^{-1} kn\right)\\
&= \Theta(kn).
\end{align*}
Comparing the two previous bounds on $\mathbb{E}[p']$, we get
$$\Theta(kn) \le \mathbb{E}[p'] \le p \cdot \Theta\left(\frac{k \eps^{-1}}{d}\right)^{k}.$$
Rearranging this inequality gives our desired inequality of
\begin{align*}
p \ge kn \cdot \Theta\left( \frac{d}{k \eps^{-1}} \right)^k. \tag*{\qedhere}
\end{align*}
\end{proof}

We are now ready to complete the proof of Theorem \ref{thm:warmupmain}.
Let $d := w(H)/n$.
If $d < 2k \eps^{-1}$, then we have $w(H) = O(kn \eps^{-1})$ and we are done.
Otherwise, if $d \ge 2k \eps^{-1}$, then we may apply the full counting lemma to say that $H$ has $kn \cdot \Omega(\eps d/k)^k$ monotone safe $k$-paths.
Meanwhile, the dispersion lemma implies that $H$ has $O(n^2)$ such paths.
Comparing these estimates, we get
$$kn \cdot \Omega\left(\frac{d}{k\eps^{-1}}\right)^k \le O(n^2).$$
Rearranging terms in this inequality to isolate $d$, we get
$$d \le O\left(\eps^{-1} k^{(k-1)/k} n^{1/k}\right) = O\left(\eps^{-1} k n^{1/k} \right),$$
and thus $w(H) = w(H \setminus \cee) + w(\cee) = nd + n = O\left( \eps^{-1} k n^{1+1/k} \right)$.

\section{Full Proof: Light Spanners via Bucket-Monotone Paths \label{sec:mainproof}}

We are now ready to prove our main theorem:

\begin{theorem} [Main Theorem] \label{thm:main}
Let $\eps>0$, let $k, n$ be positive integers, and let $H$ be an $n$-node graph with a unit-weight spanning cycle $\cee$ and weighted girth $>(1+4 \eps) \cdot 2k$.
Then
$$w(H) = O\left( \eps^{-1} n^{1+1/k} \right).$$
\end{theorem}

By the reductions in Section \ref{sec:priorwork}, this implies Theorem \ref{thm:sotalight}.
As in the warmup proof, we use $4\eps$ rather than $\eps$ in the weighted girth purely for convenience, and we will arbitrarily define a \forward{} and \backward{} direction around the spanning cycle.

\subsection{Bucket-Safe Paths and Full Proof Intuition}

Our full proof will save a factor of $k$ in lightness over the previous warmup proof.
Before proceeding to this main proof, let us give some intuition on where exactly the previous proof is losing a factor of $k$.
We will consider two extreme opposing cases for the structure of $H$, and see how our previous proof loses this $k$-factor in either one.
In the following discussion we define the \emph{bucket} $B_i \subseteq E(H)$ to be the set of edges in $H \setminus \cee$ with weights in the range $[2^i, 2^{i+1})$.

\paragraph{Special Case 1: Well-Separated Edge Weights.} Let us imagine that we are in an extreme special case in which no two edges in $H \setminus \cee$ lie in the same bucket; in particular, any two edges $e_1, e_2$ have weights $w(e_1), w(e_2)$ that differ by a factor of at least $2$.
In this case, the $k$-factor loss from our previous argument is because we are too stingy in our definition of edge-safe paths.
We allowed only $s \le \eps w(e)$ spanning cycle edges before and after each spanning cycle edge $e$, but in fact our dispersion lemma would still hold if we allowed approximately $\eps k w(e)$ such edges.
This is because the separation of edge weights implies that the total number of non-spanning-cycle edges will \emph{telescope} over the monotone path; e.g., revisiting the proof of Lemma \ref{lem:warmupdisp} there will still be only $O(\eps k w(e_j^a))$ in the cycle $C$ that we analyze.
These ultimately cost only a negligible $+O(\eps k)$ term in the bound on weighted girth.
This leads to our first piece of intuition on how to improve the argument:

\begin{center}
\emph{We can budget $O(\eps k 2^i)$ spanning cycle edges per bucket $B_i$, rather than $O(\eps w(e))$ spanning cycle edges per edge $e$.}
\end{center}

\paragraph{Special Case 2: Bucketed Edge Weights.} Let us now imagine the opposite special case, in which all edges in $H \setminus \cee$ lie in the same bucket $B_i$.
In this case, as discussed earlier, \emph{monotonicity} is not needed in the proof; we can use safe $k$-paths that use edges from $B_i$ in any order.
For these (not-necessarily-monotone) safe $k$-paths, the weak counting lemma can be improved by a factor of $k$: they are guaranteed to exist once $w(H \setminus \cee) \ge \eps^{-1} n$, rather than $w(H \setminus \cee) \ge \eps^{-1} kn$.
This leads to our second piece of intuition on how to gain a $k$-factor:
\begin{center}
\emph{The non-spanning-cycle edges within any bucket $B_i$ do not need to occur in monotone order.}
\end{center}

\paragraph{Bucket-Safe Paths.}

Our full proof will use \emph{bucket-safe paths}, which replace the \emph{edge-safe paths} from the previous warmup proof:

\begin{definition} [Bucket-Safe Paths]
A path $\pi$ in $H$ is safe for bucket $B_i$ if it is non-backtracking (meaning that it does not repeat any edge twice in a row), all of its non-spanning-cycle edges are in $B_i$, and for some integer $0 \le s \le \eps k 2^i$ it contains exactly $2s$ spanning cycle edges, where the first $s$ are used in the \forward{} direction and the last $s$ are used in the \backward{} direction.
We say that $\pi$ is extra-safe if in fact $s \le \eps k 2^{i-1}$.
\end{definition}

We note that any empty (single-node) path is considered to be safe for every bucket.
Bucket-safe paths incorporate both pieces of intuition: they allow the edges within bucket $B_i$ to occur in any order, and they budget $\eps k 2^i$ non-spanning-cycle edges for this bucket.
The following technical claim extends Claim \ref{clm:esmatching} in the natural way to bucket-safe paths, and it has essentially the same proof:
\begin{claim} \label{clm:bsmatching}
Let $q, q'$ be bucket-safe paths in $H$, let $\sigma, \sigma'$ be the sequences of non-spanning-cycle edges used by $q, q'$ respectively, and suppose that $q, q'$ share an endpoint node $y$.
If $q \ne q'$, then $\sigma \ne \sigma'$.
\end{claim}
\begin{proof}
We will prove the contrapositive.
Suppose that
$$\sigma = \sigma' =: \left((u_1, v_1), \dots, (u_j, v_j)\right).$$
For all pairs of nodes $v_i, u_{i+1}$, and also for $v_j, y$, there are two possible spanning cycle paths between these nodes (going around $\cee$ in either direction), and at least one of these paths must use $\ge n/2$ spanning cycle edges.
By Lemma \ref{lem:maxwt} the maximum edge weight in $H$ is $W < \frac{n}{4 \eps k}$, and thus any bucket-safe path uses at most $2 \eps k W < n/2$ spanning cycle edges.
Thus $q, q'$ must choose the same spanning cycle paths between each of these pairs of nodes.

This shows that $q, q'$ are identical on their suffix from $u_1$ to $y$.
In particular $q, q'$ use the same number of \forward{} and \backward{} steps following $u_1$.
By definition bucket-safe paths must use the same number of \forward{} and \backward{} steps overall, which implies that they also use the same spanning cycle paths as prefixes from their start node up to $u_1$. This implies $q=q'$.
\end{proof}

Our plan is to use bucket-safe paths as building blocks to construct \emph{bucket-monotone paths}, which are formed by concatenating bucket-safe paths in monotonic order of bucket weight, exactly like we previously constructed monotone safe paths by concatenating edge-safe paths.

\subsection{Bucket-Monotone Paths and the Dispersion Lemma}

We focus on the following paths in our proof:
\begin{definition} [Bucket-Monotone Safe $k$-Paths]
A path $\pi$ in $H$ is a bucket-monotone safe $k$-path if it has exactly $k$ non-spanning-cycle edges in total, and it can be partitioned into (possibly empty) subpaths $\pi = q_0 \circ \dots \circ q_j$, where each subpath $q_i$ is safe for bucket $B_i$ (and so these bucket weights are increasing along $\pi$).
We say that $\pi$ is extra-safe if each subpath $q_i$ is extra-safe for $B_i$.
\end{definition}

We note that, although each individual bucket-safe path $q_i$ is non-backtracking, a bucket-monotone safe $k$-path may backtrack, e.g.\ when edges from $q_{i+1}$ backtrack edges from $q_i$.
This will be used in the counting lemma.
The dispersion lemma for bucket-monotone safe $k$-paths is very similar to the one from the warmup proof.

\begin{lemma} [Dispersion Lemma] \label{lem:fulldisp}
$H$ may not have two distinct bucket-monotone safe $k$-paths with the same endpoints $s, t$.
\end{lemma}
\begin{proof}
Seeking contradiction, let $\pi^a, \pi^b$ be two distinct bucket-monotone safe $k$-paths with the same endpoints $s, t$.
Let the decomposition of $\pi^a$ into bucket-safe paths be
$\pi^a = q^a_0 \circ q^a_1 \circ \dots$, and let $\sigma^a_i$ be the (possibly empty) sequence of non-spanning-cycle edges used in $q^a_i$.
We use similar notation for $\pi^b$.

Let $j$ be the last index on which $q^a_j \ne q^b_j$.
Note that $q^a_{j}, q^b_{j}$ share an endpoint node, which we will call $y$.
Then by Claim \ref{clm:bsmatching}, we have $\sigma^a_{j} \ne \sigma^b_{j}$.
Thus the $s \leadsto y$ prefixes of $\pi^a, \pi^b$ are distinct, and in particular there exists a cycle
$$C \subseteq \pi^a[s \leadsto y] \cup \pi^b[s\leadsto y]$$
where $C$ contains at least one edge from $\sigma^a_{j} \cup \sigma^b_{j}$.
Let $e^*$ be the heaviest edge in $C$.
Our goal is now to bound $w(C)$, and it will be helpful to separately count the contribution of the spanning cycle and non-spanning-cycle edges.

\paragraph{Non-Spanning-Cycle Edges:} Since $\pi^a, \pi^b$ are bucket-safe $k$-paths, together they contain at most $2k$ non-spanning-cycle edges, each of weight $\le w(e^*)$.
So these contribute at most $2k w(e^*)$ to $w(C)$.

\paragraph{Spanning Cycle Edges:} Since each subpath $q^a_i, q^b_i$ is safe for bucket $B_i$, it contains at most $\eps k 2^{i+1}$ spanning cycle edges.
So the total number of spanning cycle edges in $\pi^a[s\leadsto y] \cup \pi^b[s\leadsto y]$ is at most
$$2 \cdot \sum \limits_{i=1}^j \eps k 2^{i+1} < \eps k 2^{j+3}.$$
Finally, since $e^*$ is in bucket $B_j$, we have $w(e^*) \ge 2^j$.
Putting the parts together, we have
$$w^*(C) = \frac{w(C)}{w(e^*)} < \frac{2k w(e^*) + \eps k 2^{j+3}}{w(e^*)} \le 2k + 8 \eps k = (1+4\eps) \cdot 2k.$$
This contradicts that $H$ has weighted girth $> (1+4\eps) \cdot 2k$, which completes the proof.
\end{proof}

\subsection{New Hiker Paths and the Counting Lemma}

Much like the warmup proof, the medium and full counting lemmas generalize easily from the Moore bounds, but a new conceptual idea is needed for the weak counting lemma.
For intuition, and for fun, we will make up an extension of the previous hiker puzzle that captures the gist of how our new weak counting lemma extends the one from the warmup proof.

\begin{puzzle}[Another Hiker Puzzle]\label{puzzle:another}
Graph National Park has $n$ landmarks (nodes), $m$ trails (undirected edges), and each trail has a difficulty rating (weight).
It is possible for several trails to receive the same difficulty rating.
The park has also installed a shuttle system, with a fleet of shuttles that drive in both the \forward{} and \backward{} direction in a loop (spanning cycle) around the landmarks.
We have a Visitor's Pass that lets us ride the shuttle for up to $2t$ stops per day, which can be split across several trips if we like.

We are planning a multi-day backpacking trip to the park.
At any time we may camp overnight at a landmark.
We have the following constraints:
\begin{itemize}
\item To avoid boredom, we must hike at least one trail per day.
We may also never \emph{backtrack} a trail, meaning that after we hike a trail, our next action cannot be to immediately hike the same trail again in the reverse direction.
(But we can otherwise repeat trails, e.g., by riding the shuttle back to the start of the trail and hiking it again.)

\item Over the entire trip, we must hike all trails in nonincreasing order of difficulty.
Additionally, each time we camp we wake up sore, and so the trails we hike on the following day must be \emph{strictly} easier than all trails hiked on the previous day.
\end{itemize}

\noindent \textbf{Prove:} If we start at the right landmark, we can hike at least $2mt/n$ trails.

\begin{center}
\begin{tikzpicture}

% Define the points for the path
\coordinate (A) at (0,0);
\coordinate (B) at (1,1);
\coordinate (B') at (3,1);
\coordinate (C) at (4,0);
\coordinate (C') at (7, 0);
\coordinate (D) at (8,2);
\coordinate (D') at (10, 2);
\coordinate (E) at (12,1);

\node [left] at (A) {\Strichmaxerl[3]};
\node [above] at (B) {\Summertree[3]};
\node [below=0.6, align=center] at (C) {\Huge \color{brown} \VarMountain};
\node [above] at (D) {\Huge \color{red} \VarFlag};
\node [right] at (E) {\Huge \color{green!50!black} \Tent};
\node [above] at (B') {\Huge \color{blue} \FilledHut};
\node [below=0.2, align=center] at (C') {\Bed \\ sleep};
\node [above, align=center] at (D') {\Bed \\ sleep};

% Draw the lines
\draw [thick] (A) -- node[above, sloped] {12} (B);
\draw [thick] (B') -- node[above, sloped] {12} (C);
\draw [thick] (C') -- node[above, sloped] {9} (D);
\draw [thick, ->] (D') -- node[above, sloped] {4} (E);

% draw bus lines
\draw [snake it] (B) -- node [below, align=center] {\faBus \\ $s_1$ stops} (B');
\draw [snake it] (C) -- node [below, align=center] {\faBus \\ $s_2$ stops} (C');
\draw [snake it] (D) -- node [below, align=center] {\faBus \\ $s_3$ stops} (D');

% Optionally, draw the points
\fill (A) circle[radius=2pt];
\fill (B) circle[radius=2pt];
\fill (C) circle[radius=2pt];
\fill (D) circle[radius=2pt];
\fill (B') circle[radius=2pt];
\fill (C') circle[radius=2pt];
\fill (D') circle[radius=2pt];
%\fill (E) circle[radius=2pt];

\end{tikzpicture}

~

A valid hiker path, $4$ trails long and split over 3 days, assuming $s_1 + s_2 \le 2t$ and $s_3 \le 2t$.
\end{center}
\end{puzzle}

\vfill
\begin{center}
  (This space is intentionally left free to not spoil the solution to Puzzle~\ref{puzzle:another}.)
\end{center}
\vfill
\clearpage

\begin{solution}
As in our previous hiker puzzle, we assemble a squad of $n$ hikers, and we start with one hiker standing at each landmark.
We plan out hiker paths by the following process.
\begin{itemize}
\item First, we group the trails by difficulty rating.
Let $B_0$ be all the trails that are tied for highest difficulty rating, let $B_1$ be the trails that are tied for second-highest difficulty rating, etc.
We will plan hiker paths for each day $i$ using the following steps:

\begin{itemize}
\item \textbf{(Dawn)} At dawn, each hiker plans a path that consists of riding the shuttle $t$ stops in the \forward{} direction from their current position (but they do not yet walk their planned path).
Notice an \textbf{invariant}: for each node $v$, and for each integer $1 \le s \le t$, there exists a hiker who reaches $v$ as the endpoint of their $s^{th}$ shuttle stop.
As we modify paths in the following step, this invariant will remain.

\item \textbf{(Morning)} Next, the hikers collectively consider the trails $(u, v) \in B_i$ one at a time, in arbitrary order.
For each integer $1 \le s \le t$, consider the two hikers $h_u, h_v$ who currently plan to reach the nodes $u, v$ as their $s^{th}$ shuttle stop.
We insert $(u, v)$ into the paths of these two hikers right after their $s^{th}$ shuttle stops.
Following $(u, v)$, the two hikers swap paths: $h_u$ takes $h_v$'s previous path from $v$ onward, and $h_v$ takes $h_u$'s previous path from $u$ onward.

\item \textbf{(Afternoon)} Finally, in the afternoon, the hikers all travel the paths they have planned out.
Then they end the day by riding the shuttle $t$ steps \backward{} and camping at the landmark they reach.
\end{itemize}
\end{itemize}

\begin{center}
  \begin{minipage}{.3\textwidth}
    	\includegraphics[width=\textwidth]{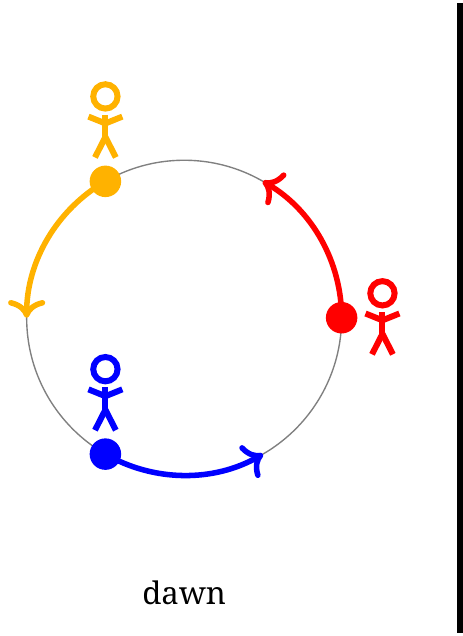}
\end{minipage}
\hfill
\begin{minipage}{.3\textwidth}
  	\includegraphics[width=\textwidth]{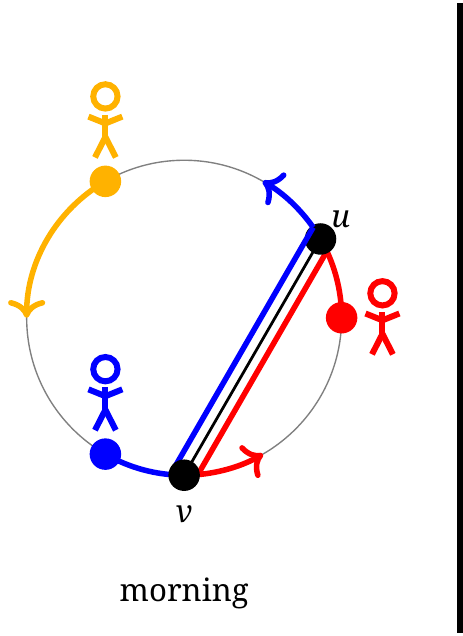}
\end{minipage}
\hfill
\begin{minipage}{.3\textwidth}
  	\includegraphics[width=\textwidth]{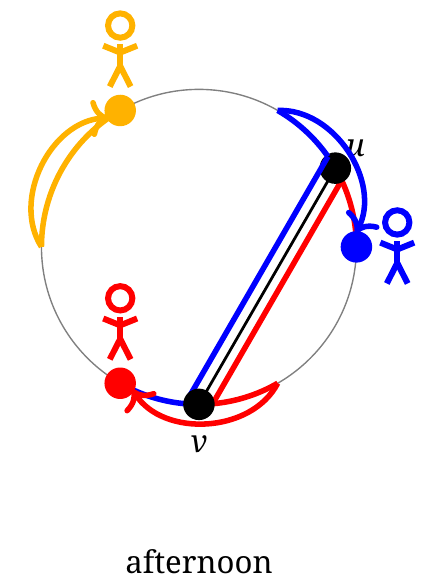}
\end{minipage}
\end{center}

Each trail $(u, v)$ is inserted into the paths of $2t$ hikers.
% We also notice that no hiker will \emph{backtrack} $(u, v)$.
% That is: say a hiker hikes $(u, v)$ because $u$ is their $s^{th}$ shuttle stop.
% If they were to later hike $(u, v)$ again in the opposite direction, it would have to be because $v$ is their $s+1^{st}$ or later shuttle stop.
% That means that the hiker must have taken the shuttle at least once (and also hiked some trails) in between their trips across $(u, v)$.
So our $n$ hikers hike $2mt$ trails in total, and so there is a hiker who hikes at least $2mt/n$ trails.
There is one last detail: we are supposed to ensure that this hiker hikes at least one trail per day.
Notice that, if a hiker does not hike any trails on day $i$, then their planned path for the day consists of riding the shuttle $t$ stops \forward{}, riding the shuttle $t$ stops \backward{}, and camping back at the same landmark where their day began.
The hiker can therefore skip this day entirely in their itinerary, and so they hike at least one trail in each un-skipped day.
\end{solution}

~\\We now give our weak counting lemma, which is more or less a repeat of the hiker path protocol from the above puzzle, with some adaptations and slightly more formal exposition of some steps.

\begin{lemma} [Weak Counting Lemma]
If $w(H \setminus \cee) \ge 4 \eps^{-1} n$, then $H$ contains a bucket-monotone extra-safe $k'$-path, for some $k' \ge k$.
\end{lemma}
\begin{proof}
Start by placing a hiker at each node of $H$.
We will define a protocol to generate paths for these hikers, which loops on three phases; each loop is called a \textbf{day}.

\begin{itemize}
\item \textbf{(Dawn)} On day $i$ at dawn, each hiker plans a path that consists of walking $t_i :=\lfloor \eps k 2^{i-1} \rfloor$ steps \forward{} on the spanning cycle (but they do not walk these edges yet).
We notice that the following invariant holds:
\begin{center}
\textbf{Invariant}: for each node $v$, for each integer $1 \le s \le t_i$ there exists a hiker who plans to reach $v$ as the endpoint of the $s^{th}$ spanning cycle edge that they walk, and also there exists a hiker who plans to end their journey at $v$.
\end{center}

\item \textbf{(Morning)} On day $i$ in the morning, the hikers collectively consider the edges $(u, v) \in B_i$ one at a time, in arbitrary order.
For each integer $1 \le s \le t_i$, by the invariant, we may let $h_{u}, h_{v}$ be the two hikers who respectively reach $u, v$ at the end of their $s^{th}$ \forward{} spanning cycle steps.
We insert the edge $(u, v)$ into the planned paths of these two hikers immediately after their $s^{th}$ spanning cycle steps.
Following this edge, these hikers plan to swap paths; $h_{u}$ takes the path suffix previously planned by $h_{v}$ from $v$ onward, and $h_{v}$ takes the path suffix previously planned by $h_{u}$ from $u$ onward.

Note that each suffix swap preserves the invariant, since each $(s+1), \dots, t_i$ spanning cycle edge used by $h_u$ becomes the $(s+1), \dots, t_i$ spanning cycle edge used by $h_v$ (and vice versa), and the node at which $h_u$ plans to end their journey becomes the node at which $h_v$ plans to end their journey (and vice versa).

\item \textbf{(Afternoon)} Finally, for each hiker $h$, let $f_h$ be the number of contiguous \forward{} spanning cycle edges that form a suffix of their planned path.
We replace this suffix with $t_i-f_h$ \backward{} steps on the spanning cycle.
We can think of this step as adding a suffix of $t_i$ \backward{} steps on the spanning cycle to the end of each hiker's path, and then repeatedly canceling adjacent \forward{} and \backward{} steps that backtrack each other.

Each hiker then walks their planned path, and then ends the day; we then proceed to day $i+1$ and continue until all buckets have been considered.
\end{itemize}

For correctness of this hiker protocol, we note that by the invariant, it remains true that one hiker stands at each node at the end of the day and thus the invariant holds again in the following day at dawn.
This construction implies that the path hiked by each hiker on day $i$ is extra-safe for $B_i$.
Thus, each hiker's overall journey (concatenating their path across each day) forms a bucket-monotone extra-safe path.
The number of hikers who traverse an edge $(u, v) \in B_i$ is
$$2t = 2 \lfloor \eps k 2^{i-1} \rfloor \ge \frac{\eps k 2^{i+1}}{4} \ge \eps k \cdot \frac{w(u, v)}{4}.$$
Summing over all edges in $H \setminus \cee$, the total number of non-spanning-cycle edges traversed by our $n$ hikers is therefore at least
$$\sum \limits_{(u, v) \in E(H \setminus \cee)} \eps k \cdot \frac{w(u, v)}{4} = \eps k \cdot \frac{w(H \setminus \cee)}{4} \ge kn.$$
So there exists a hiker who hikes a bucket-monotone extra-safe $k'$-path, for some $k' \ge k$.
\end{proof}

The medium and full counting lemmas can now be bootstrapped from the weak counting lemma by the same method as in our previous proofs.

\begin{lemma} [Medium Counting Lemma]
$H$ contains at least $\Theta(\eps) \cdot \left( w(H \setminus \cee) - 4 \eps^{-1} n \right)$ bucket-monotone safe $k$-paths.
\end{lemma}
\begin{proof}
Repeat the following process until no longer possible.
Find a bucket-monotone extra-safe $k'$-path $\pi$, for some $k' \ge k$.
We can use $\pi$ to record a family of bucket-monotone (not-necessarily-extra-)safe $k$-paths by the following process:
\begin{itemize}
\item Let $e_1, \dots, e_k$ be the first $k$ non-spanning-cycle edges used by $\pi$.
Truncate $\pi$ immediately after $e_k$.
After this truncation, the last path $q_j$ in the decomposition might not end with a suffix of the appropriate number of \backward{} edges, and so we add the appropriate number of \backward{} edges as a suffix to $\pi$ to restore the fact that $q_j$ is extra-safe for $B_j$.

\item Let $B_i$ be the bucket that contains $e_1$, and let $B_j$ be the bucket that contains $e_k$.
Thus, omitting empty paths at the beginning and end, we may write the decomposition of $\pi$ into bucket-extra-safe paths as $\pi = q_i \circ \dots \circ q_j$, where each subpath $q_x$ is extra-safe for bucket $B_x$.

\item For each integer $0 \le s \le \lfloor \eps k 2^{i-1} \rfloor$, notice that we can generate a bucket-monotone safe $k$-path by modifying $\pi$ by adding a prefix of $s$ additional \forward{} spanning cycle edges, and a suffix of $s$ additional \backward{} spanning cycle edges, to each nonempty path $q_x$ in the decomposition of $\pi$.
\end{itemize}

We record $\lfloor \eps k 2^{i-1} \rfloor = \Theta(\eps k w(e_1))$ bucket-monotone safe $k$-paths in this way.
We then delete $e_1$ to ensure that we do not re-record these paths in a future round.
By the weak counting lemma, we can repeat this process at least until $w(H \setminus \cee) < 4 \eps^{-1} n$.
So we record at least $\Theta(\eps) \cdot (w(H \setminus \cee) - 4 \eps^{-1} n)$ paths in total.
\end{proof}

In the full counting lemma for our warmup proof, we used the fact that monotone safe $k$-paths necessarily contain $k$ \emph{distinct} non-spanning cycle edges, by monotonicity.
For our new full counting lemma we will need the analogous fact is also true for bucket-monotone safe $k$-paths, as established in the following claim:
\begin{claim} \label{clm:bmsdistinct}
Every bucket-monotone safe $k$-path $\pi$ uses $k$ distinct non-spanning-cycle edges.
\end{claim}
\begin{proof}
Suppose for contradiction that there is a repeated non-spanning-cycle edge $(u, v)$ in $\pi$.
Let $B_i$ be the bucket that contains $(u, v)$, and let $q_i$ be the subpath of $\pi$ that is safe for $B_i$.
Notice that $q_i$ must contain both copies of $(u, v)$, and so it has a cycle $C$ as a subpath between these copies.
Let $e^*$ be the heaviest edge in $C$.
We may bound the cycle weight $w(C)$ as follows.
First, there are at most $k$ non-spanning-cycle edges in $C$, and so the total contribution of these edges to $w(C)$ is at most
$$\sum \limits_{(u, v) \text{ non-sp-cyc edge in } C} w(u, v) \le k \cdot w(e^*).$$
Meanwhile, the number of spanning cycle edges in $C$ is at most $\eps k 2^i$.
Putting these together, we have
$$w^*(C) \le \frac{w(C)}{w(e^*)} \le \frac{k \cdot w(e^*) + \eps k 2^i}{w(e^*)} \le k + 2 \eps k = (1 + 2\eps) k.$$
This contradicts the weighted girth of $H$.
\end{proof}

\begin{lemma} [Full Counting Lemma]
Let $d := w(H \setminus \cee) / n$.
If $d \ge 5 \eps^{-1}$, then $H$ contains at least $n \cdot \Omega(\eps d)^{k}$ bucket-monotone safe $k$-paths.
\end{lemma}
\begin{proof}
Let $H'$ be a random edge-subgraph of $H$ obtained by keeping the spanning cycle $\cee$ deterministically, and keeping each non-spanning-cycle edge independently with probability $5\eps^{-1} / d$.
Let $p, p'$ be the number of bucket-monotone safe $k$-paths in $H, H'$ respectively.

On one hand, by Claim \ref{clm:bmsdistinct}, every such path $\pi$ uses $k$ \emph{distinct} non-spanning-cycle edges.
Thus, the probability that $\pi$ survives in $H'$ is $\Theta(\eps^{-1} / d)^{k}$, and so
$$\mathbb{E}[p'] = p \cdot \Theta\left( \frac{\eps^{-1}}{d} \right)^{k}.$$
On the other hand, we have

\begin{align*}
\mathbb{E}[p'] &\ge \mathbb{E}\left[\Theta(\eps) \cdot \left( w(H' \setminus \cee) - 4\eps^{-1} n \right)\right] \tag*{Medium Counting Lemma}\\
&= \Theta(\eps) \cdot \left(\mathbb{E}\left[ w(H' \setminus \cee)\right] - 4\eps^{-1} n\right)\\
&= \Theta(\eps) \cdot \left( w(H \setminus \cee) \cdot \frac{5 \eps^{-1}}{d} - 4\eps^{-1}n\right)\\
&= \Theta(\eps) \cdot \left( 5 \eps^{-1} n - 4\eps^{-1} n\right)\\
&= \Theta(\eps) \cdot \left( \eps^{-1} n\right)\\
&= \Theta(n).
\end{align*}
Comparing the two previous bounds on $\mathbb{E}[p']$, we get
$$\Theta(n) \le \mathbb{E}[p'] \le p \cdot \Theta\left(\frac{\eps^{-1}}{d}\right)^{k}.$$
Rearranging now gives our desired inequality of
$p \ge n \cdot \Theta\left( \eps d \right)^k$.
\end{proof}

We now complete the proof of Theorem \ref{thm:main} in the usual way.
Let $d := w(H \setminus \cee)/n$.
If $d < 5 \eps^{-1}$, then we have $w(H) = O(\eps^{-1} n)$ and we are done.
Otherwise, by the dispersion and full counting lemmas, the number of bucket-monotone safe $k$-paths in $H$ is at least $n \cdot \Omega(\eps d)^k$, and at most $O(n^2)$.
So we have
\begin{align*}
n \cdot \Omega(\eps d)^k &\le O(n^2)\\
d &\le O\left(\eps^{-1} n^{1/k}\right)
\end{align*}
and thus $w(H) = w(H \setminus \cee) + w(\cee) = nd + n = O\left(\eps^{-1} n^{1+1/k}\right)$.

\section*{Acknowledgments}

I am grateful to Michael Elkin for helpful technical comments on an earlier draft of this paper.

\printbibliography

@inproceedings{Bodwin24,
	author = {Greg Bodwin},
	bibsource = {dblp computer science bibliography, https://dblp.org},
	booktitle = {2024 Symposium on Simplicity in Algorithms, {SOSA} 2024, Alexandria, VA, USA, January 8-10, 2024},
	doi = {10.1137/1.9781611977936.5},
	editor = {Merav Parter and Seth Pettie},
	pages = {39--55},
	publisher = {{SIAM}},
	title = {An Alternate Proof of Near-Optimal Light Spanners},
	url = {https://doi.org/10.1137/1.9781611977936.5},
	year = {2024}}

@article{ADFSW22,
	author = {Alstrup, Stephen and Dahlgaard, S{\o}ren and Filtser, Arnold and St{\"o}ckel, Morten and Wulff-Nilsen, Christian},
	journal = {Theoretical Computer Science},
	pages = {82--112},
	publisher = {Elsevier},
	title = {Constructing light spanners deterministically in near-linear time},
	volume = {907},
	year = {2022},
  doi = {10.1016/j.tcs.2022.01.021}
}

@article{BF24,
	author = {Bodwin, Greg and Flics, Jeremy},
	journal = {arXiv preprint arXiv:2406.04459},
	title = {A Lower Bound for Light Spanners in General Graphs},
    doi={10.48550/arXiv.2406.04459},
	year = {2024}}

@book{CLRS22,
	author = {Cormen, Thomas H and Leiserson, Charles E and Rivest, Ronald L and Stein, Clifford},
	publisher = {MIT press},
	title = {Introduction to algorithms},
	year = {2022}}

@article{ES16,
	author = {Elkin, Michael and Solomon, Shay},
	journal = {ACM Transactions on Algorithms (TALG)},
	number = {3},
	pages = {1--21},
	publisher = {ACM New York, NY, USA},
	title = {Fast constructions of lightweight spanners for general graphs},
	volume = {12},
          doi = {10.1145/2836167},
	year = {2016}}

@inproceedings{BLW19,
	author = {Borradaile, Glencora and Le, Hung and Wulff-Nilsen, Christian},
	booktitle = {Proceedings of the Thirtieth Annual ACM-SIAM Symposium on Discrete Algorithms, {SODA} 2019},
	organization = {SIAM},
	pages = {2371--2379},
	title = {Greedy spanners are optimal in doubling metrics},
        doi = {10.1137/1.9781611975482.145},
	year = {2019}}

@inproceedings{BLW17,
	author = {Borradaile, Glencora and Le, Hung and Wulff-Nilsen, Christian},
	booktitle = {2017 IEEE 58th Annual Symposium on Foundations of Computer Science, {FOCS} 2017},
	organization = {IEEE},
	pages = {767--778},
	title = {Minor-free graphs have light spanners},
        doi = {10.1109/FOCS.2017.76},
	year = {2017}}

@article{CG09,
	author = {Chan, T-H Hubert and Gupta, Anupam},
	journal = {Discrete \& Computational Geometry},
	pages = {28--44},
	publisher = {Springer},
	title = {Small hop-diameter sparse spanners for doubling metrics},
	volume = {41},
      doi = {10.1007/s00454-008-9115-5},
	year = {2009}}

@article{ES15,
	author = {Elkin, Michael and Solomon, Shay},
	journal = {Journal of the ACM (JACM)},
	number = {5},
	pages = {1--45},
	publisher = {ACM New York, NY, USA},
	title = {Optimal euclidean spanners: Really short, thin, and lanky},
	volume = {62},
  doi = {10.1145/2819008},
	year = {2015}}

@article{CLNS15,
	author = {Chan, T-H Hubert and Li, Mingfei and Ning, Li and Solomon, Shay},
	journal = {SIAM Journal on Computing},
	number = {1},
	pages = {37--53},
	publisher = {SIAM},
	title = {New doubling spanners: Better and simpler},
	volume = {44},
  doi = {10.1137/130930984},
	year = {2015}}

@article{LT24,
	author = {Le, Hung and Than, Cuong},
	journal = {ACM Transactions on Algorithms},
	number = {3},
	pages = {1--30},
	publisher = {ACM New York, NY},
	title = {Greedy spanners in euclidean spaces admit sublinear separators},
	volume = {20},
        doi = {10.1145/3590771},
	year = {2024}}

@inproceedings{KLMS22,
	author = {Kahalon, Omri and Le, Hung and Milenkovi{\'c}, Lazar and Solomon, Shay},
	booktitle = {Proceedings of the 2022 ACM Symposium on Principles of Distributed Computing, {PODC} 2022},
	pages = {151--162},
	title = {Can't see the forest for the trees: Navigating metric spaces by bounded hop-diameter spanners},
        doi = {10.1145/3519270.3538414},
	year = {2022}}

@book{NS07,
	author = {Narasimhan, Giri and Smid, Michiel},
	publisher = {Cambridge University Press},
	title = {Geometric spanner networks},
        doi = {10.1017/CBO9780511546884},
	year = {2007}}

@article{BKKLLPT24,
	author = {Bhore, Sujoy and Keszegh, Bal{\'a}zs and Kupavskii, Andrey and Le, Hung and Louvet, Alexandre and P{\'a}lv{\"o}lgyi, D{\"o}m{\"o}t{\"o}r and T{\'o}th, Csaba D},
	journal = {arXiv preprint arXiv:2404.05045},
	title = {Spanners in Planar Domains via Steiner Spanners and non-Steiner Tree Covers},
    doi={10.48550/arXiv.2404.05045},
	year = {2024}}

@article{FS20,
	author = {Filtser, Arnold and Solomon, Shay},
	journal = {SIAM Journal on Computing},
	number = {2},
	pages = {429--447},
	publisher = {SIAM},
	title = {The greedy spanner is existentially optimal},
	volume = {49},
        doi = {10.1137/18M1210678},
	year = {2020}}

@inproceedings{girth,
	author = {Erd{\H{o}}s, Paul},
	booktitle = {Proceedings of the Symposium on Theory of Graphs and its Applications},
	pages = {2936},
	title = {Extremal problems in graph theory},
	year = {1963}}

@article{PU89jacm,
	author = {Peleg, David and Upfal, Eli},
	journal = {Journal of the ACM (JACM)},
	number = {3},
	pages = {510--530},
	publisher = {ACM},
	title = {A trade-off between space and efficiency for routing tables},
	volume = {36},
        doi = {10.1145/65950.65953},
	year = {1989}}

@article{PU89sicomp,
	author = {Peleg, David and Ullman, Jeffrey},
	journal = {SIAM Journal on Computing (SICOMP)},
	number = {4},
	pages = {740---747},
	publisher = {SIAM},
	title = {An Optimal Synchronizer for the Hypercube},
	volume = {18},
        doi = {10.1137/0218050},
	year = {1989}}

@article{ADDJS93,
	author = {Alth{\"o}fer, Ingo and Das, Gautam and Dobkin, David and Joseph, Deborah and Soares, Jos{\'e}},
	journal = {Discrete \& Computational Geometry},
	number = {1},
	pages = {81--100},
	publisher = {Springer},
	title = {On sparse spanners of weighted graphs},
	volume = {9},
  doi = {10.1007/BF02189308},
	year = {1993}}

@inproceedings{GR08,
	author = {Gottlieb, Lee-Ad and Roditty, Liam},
	booktitle = {European Symposium on Algorithms},
	organization = {Springer},
	pages = {478--489},
	title = {An optimal dynamic spanner for doubling metric spaces},
    doi = {10.1007/978-3-540-87744-8_40},
	year = {2008}}

@article{CW18,
  title={Near-optimal light spanners},
  author={Chechik, Shiri and Wulff-Nilsen, Christian},
  journal={ACM Transactions on Algorithms (TALG)},
  volume={14},
  number={3},
  pages={1--15},
  year={2018},
  doi={10.1145/3199607},
  publisher={ACM New York, NY, USA}
}

@inproceedings{LS23,
	author = {Hung Le and Shay Solomon},
	booktitle = {Proceedings of the 55th Annual ACM SIGACT Symposium on Theory of Computing, {STOC} 2023},
	organization = {ACM},
	title = {A Unified Framework for Light Spanners},
        doi = {10.1145/3564246.3585185},
	year = {2023}}

@article{ENS15,
  title={Light spanners},
  author={Elkin, Michael and Neiman, Ofer and Solomon, Shay},
  journal={SIAM Journal on Discrete Mathematics},
  volume={29},
  number={3},
  pages={1312--1321},
  year={2015},
  doi={10.1137/140979538},
  publisher={SIAM}
}

@article{CDNS95,
  author       = {Barun Chandra and
                  Gautam Das and
                  Giri Narasimhan and
                  Jos{\'{e}} Soares},
  title        = {New sparseness results on graph spanners},
  journal      = {Int. J. Comput. Geom. Appl.},
  volume       = {5},
  pages        = {125--144},
  year         = {1995},
  url          = {https://doi.org/10.1142/S0218195995000088},
  doi          = {10.1142/S0218195995000088},
  bibsource    = {dblp computer science bibliography, https://dblp.org}
}

@article{CLN15,
	author = {Chan, T-H Hubert and Li, Mingfei and Ning, Li},
	journal = {Algorithmica},
	number = {1},
	pages = {53--65},
	publisher = {Springer},
	title = {Sparse fault-tolerant spanners for doubling metrics with bounded hop-diameter or degree},
	volume = {71},
      doi = {10.1007/s00453-013-9779-y},
	year = {2015}}

@article{ABSHJKS20,
	author = {Ahmed, Reyan and Bodwin, Greg and Sahneh, Faryad Darabi and Hamm, Keaton and Jebelli, Mohammad Javad Latifi and Kobourov, Stephen and Spence, Richard},
	journal = {Computer Science Review},
	pages = {100253},
	publisher = {Elsevier},
	title = {Graph spanners: A tutorial review},
	volume = {37},
        doi={10.1016/j.cosrev.2020.100253},
	year = {2020}}

\end{document}